\documentclass[12pt,a4paper]{article}
 
\usepackage{microtype}

\usepackage[latin1]{inputenc}

\usepackage{mathrsfs,amssymb,amsthm}
\usepackage{amsmath} 
\usepackage{multicol,caption} 
\usepackage{wrapfig} 
\usepackage{fullpage,authblk}

\usepackage{subfigure}
\usepackage{url}
\usepackage{hyperref} 

\usepackage[ruled,noend,linesnumbered]{algorithm2e}

\usepackage{tikz}
 \usetikzlibrary{calc}
 \usetikzlibrary{patterns}
\usepgflibrary{shapes} 
\usepgflibrary{shapes.misc} 
\usepgflibrary{shapes.geometric} 
\usepgflibrary{patterns}
\usetikzlibrary{fit}
\usetikzlibrary{positioning}
 \newcommand{\suchthat}{\,\big|\,}
\usepackage{pifont}

 \newcommand{\bline}{\vspace{1\baselineskip}} 

    \newcommand{\splitnode}[4][1]{
      \path
      let \p{w}=(#4.west), \p{e}=(#4.east),
 \n{h}={(\x{e} - \x{w})} in 
node[draw,#2,inner sep=0,above=#1pt,at=(#4.center),anchor=south,semicircle,minimum
width=\n{h}] (#4Deux) {}
node[draw,#3,inner sep=0,above=-#1pt,at=(#4.center),anchor=north,shape border rotate=180,semicircle,minimum
width=\n{h}] (#4Un) {}
;
}

    \newcommand{\splitnodend}[4][1]{
      \path
      let \p{w}=(#4.west), \p{e}=(#4.east),
\n{h}={(\x{e} - \x{w}) - 1pt} in 
node[#2,inner sep=0,above=#1pt,at=(#4.center),anchor=south,semicircle,minimum
width=\n{h}] (#4Deux) {}
node[#3,inner sep=0,above=-#1pt,at=(#4.center),anchor=north,shape border rotate=180,semicircle,minimum
width=\n{h}] (#4Un) {}
;
}

\theoremstyle{plain}

\newtheorem{definition}{Definition}

\newtheorem{theorem}[definition]{Theorem}
\newtheorem{lemma}[definition]{Lemma}

\newtheorem{proposition}[definition]{Proposition}
\theoremstyle{remark}

 \DeclareMathOperator{\splitBlock}{Split}
 \DeclareMathOperator{\remove}{Remove}
 \DeclareMathOperator{\notRel}{NotRel}
 \DeclareMathOperator{\pre}{pre}

 \DeclareMathOperator{\initRefine}{InitRefine}

 \DeclareMathOperator{\Copy}{copy}

  \SetKwData{NotInRemove}{notInRemove} 
  \SetKwData{InRemove}{inRemove}
  \SetKwData{Rel}{Rel}
  \SetKwData{RelCount}{RelCount}
  \SetKwData{SplitCount}{splitCount}
  \SetKwData{NotRel}{NotRel}
  \SetKwData{Block}{block}
  \SetKwData{State}{state}
  \SetKwData{Node}{node}
  \SetKwData{True}{true}
  \SetKwData{False}{false}
  \SetKwData{Mark}{mark}
  \SetKwData{Count}{count}
  \SetKwData{PreB}{PreB}
  \SetKwData{Post}{post}
  \SetKwData{Pre}{pre}
  \SetKwData{PreRel}{PreRel}
  \SetKwData{Remove}{Remove}
  \SetKwData{Src}{src}
  \SetKwData{Dst}{dst}
  \SetKwData{Sl}{sl}
  \SetKwData{Label}{label}
  \SetKwData{Range}{range}
  \SetKwData{PosQp}{posQp}
  \SetKwData{Idx}{Idx$_P$}
  \SetKwData{Start}{start}
  \SetKwData{End}{end}
  \SetKwData{Index}{index}  
  \SetKwData{PreSmal}{Pre} 
  
  \SetKwFunction{Put}{put}
  \SetKwFunction{Get}{get}
  \SetKwFunction{Contains}{contains}
  \SetKwFunction{Newnode}{newNode()}
  \SetKwFunction{Newblock}{newBlock()}
  \SetKwFunction{IsEmpty}{is\_empty()}
  \SetKwFunction{States}{states}
  \SetKwFunction{MoveToHead}{moveToHead}
  \SetKwFunction{Add}{add}
  \SetKwFunction{MoveHeadTo}{moveHeadTo}
  \SetKwFunction{Init}{Init}
  \SetKwFunction{Split}{Split}

  \SetKwFunction{Take}{take}
  \SetKwFunction{Add}{add}
  \SetKwFunction{Swap}{swap}
  
  \SetKwInput{Assert}{Assert}

\title{Bisimulations over DLTS in O(m.log n)-time}

\author{G\'erard C\'ec\'e}
\affil{\small{}FEMTO-ST, UMR 6174,\\
  1 cours Leprince-Ringuet,  BP 21126,\\
  25201 Montb\'eliard Cedex France\\    
  \texttt{Gerard.Cece@femto-st.fr}}

 \date{}

\begin{document}
\maketitle

\begin{abstract}
The well known Hopcroft's algorithm to minimize deterministic complete
automata runs in $O(kn\log n)$-time, where $k$ is the size of the alphabet
and $n$ the number of states. The main part of this algorithm corresponds
to the computation of a coarsest bisimulation over a finite Deterministic
Labelled Transition System (DLTS). By applying techniques we have developed
in the case of simulations, we design a new algorithm which computes the
coarsest bisimulation over a finite DLTS in $O(m\log n)$-time and
$O(k+m+n)$-space, with $m$ the number of transitions. The underlying
DLTS does not need to be complete and thus: $m\leq kn$. This new
algorithm is much simpler than the two others found in the literature.

\end{abstract}

\section{Introduction}
\label{sec:introduction}

Bisimulation and simulation equivalences are behavioral relations between
processes. They are mainly used to reduce the state space of models since
they preserve branching time temporal logic like CTL and CTL* for
bisimulation \cite{CE81} and their existential fragment for
simulation \cite{GL94}. Simulation can also be used as a sufficient
condition for the inclusion of languages when this test is undecidable in
general \cite{CG11}.

Let us call \emph{coarsest bisimulation problem}, the problem of finding
the coarsest bisimulation equivalence in a Labelled Transition System (LTS)
and \emph{coarsest simulation problem} the corresponding problem for
simulation. One can 
establish a hierarchy, with increasing  
difficulty, from the coarsest bisimulation problem in a finite 
deterministic LTS, to the coarsest simulation problem in a finite non
deterministic 
LTS. The first case was efficiently solved by Hopcroft, to minimize finite
deterministic automata, in \cite{Hop71} with an
algorithm whose complexities are $O(kn\log n)$-time and
$O(kn)$-space \cite{Knu01}, with $k$ the size of the alphabet and $n$ the
number of states.
The next important step, the coarsest bisimulation problem for finite non
deterministic LTS, was partially solved by Paige and Tarjan in
\cite{PT87}. We say partially because it was for $k=1$. The complexities of
their algorithm are  $O(m\log n)$-time and
$O(m+n)$-space with $m$ the number of transitions. An extension to the
general case
$k>1$ was proposed later by Fernandez
\cite{Fern89}. According to \cite{VL08} the time complexity of the algorithm
of Fernandez is $O(kn\log n)$. 

In the present paper, by using techniques we have developed for the coarsest
simulation problem in \cite{Cec12}, we design now a new algorithm to avoid the
$k$ factor in the time complexity of the coarsest bisimulation problem over finite
DLTS. We do this in order to obtain a time complexity of $O(m\log
n)$ and a space complexity of $O(k+n+m)$.

In the literature, there are two papers which 
achieve the same result over finite DLTS, both are designed to minimize
finite deterministic automata. The first one
\cite{VL08}, by Valmari and Lehtinen, may be considered as the best
solution using the ideas of Hopcroft, in the sense that conceptually their
splitters are couples made of a block and of a letter. However, they need, beside the traditional
partition of the states during the computation, a partition of the set of
transitions. The second one, by B\'eal and 
Crochemore \cite{BC08} is closer to our solution but is more complex with
its use of ``all but the largest strategy'' and ``signature of states''.

\section{Preliminaries}
\label{sec:preliminaries}

Let $Q$ be a set of elements. The number of elements of $Q$ is denoted $|Q|$.
A \emph{binary relation} on $Q$ is a subset
of $Q\times Q$.
In the remainder of this paper, we consider only binary
relations, therefore when we write ``relation'' read ``binary relation''. Let
$\mathscr{R}$ be a relation on $Q$.
For $X,Y\subseteq Q$, we note $X\,\mathscr{R}\,Y$
to express the existence of two states $q,q'\in Q$ such that
$(q,q')\in X\times Y\cap\mathscr{R}$. By abuse of notation, we also note 
$q\,\mathscr{R}\,Y$ for $\{q\}\,\mathscr{R}\,Y$,  $X\,\mathscr{R}\,q'$ for
$X\,\mathscr{R}\,\{q'\}$ and $q\,\mathscr{R}\,q'$ for
$\{q\}\mathrel{\mathscr{R}}\{q'\}$. In the figures we draw  $X \mathbin{\tikz[baseline] \draw[dashed,->] (0pt,.8ex) --
  node[font=\scriptsize,fill=white,inner sep=2pt] {$\mathscr{R}$} (8ex,.8ex);}
 Y$ for $X\,\mathscr{R}\,Y$.
We note $\mathscr{R}^{-1}$ the \emph{inverse} of $\mathscr{R}$ such that
$q\mathbin{\mathscr{R}^{-1}}q'$ iff $q'\,\mathscr{R}\,q$.
We define 
$\mathscr{R}(q)\triangleq\{q'\in Q\suchthat q\,\mathscr{R}\,q'\}$ for $q\in Q$ and
$\mathscr{R}(X)\triangleq\cup_{q\in X}\mathscr{R}(q)$ for $X\subseteq Q$.
Let $\mathscr{S}$ be another relation on $Q$, the
\emph{composition} of $\mathscr{R}$ by $\mathscr{S}$ is 
$\mathscr{S}\mathrel{\circ}\mathscr{R}\triangleq\{(x,y)\in Q\times Q\suchthat
y\in\mathscr{S}(\mathscr{R}(x))\}$.
The relation $\mathscr{R}$ is said \emph{reflexive} if for all $q\in Q$, we
have $q\,\mathscr{R}\,q$.
The relation $\mathscr{R}$ is said \emph{symmetric} if
$\mathscr{R}^{-1}=\mathscr{R}$. 
The relation $\mathscr{R}$ is said \emph{transitive} if
$\mathscr{R}\mathrel{\circ}\mathscr{R}\subseteq\mathscr{R}$.
An \emph{equivalence relation} is a reflexive, symmetric and transitive
relation. For an equivalence relation $\mathscr{R}$ on $Q$ and $q\in Q$ we
call $\mathscr{R}(q)$ a \emph{block}, the block of $q$. It is well known,
and immediate consequences of the 
definitions, that for any 
equivalence relation $\mathscr{R}$ and any block $B$ 
of $\mathscr{R}$, we have $\forall q,q'\in B\,.\,q\,\mathscr{R}\,q'$,
and for any $X\subseteq Q$, $\mathscr{R}(X)$ is a, disjoint, union of blocks.

Let $X$ be a
set of subsets of $Q$, we note $\cup X\triangleq \cup_{B\in X}B$.
A \emph{partition} of $Q$ is a set of non empty subsets of $Q$, also 
called \emph{blocks}, that are pairwise disjoint and whose union gives $Q$.
There is a duality between a partition of $Q$ and an equivalence relation
on $Q$ since from a partition $P$ we can derive the equivalence relation
$\mathscr{R}_P=\cup_{B\in P}B\times B$ and from an equivalence relation
$\mathscr{R}$ we can derive a partition
$P_\mathscr{R}=\cup\{\mathscr{R}(q)\subseteq Q\suchthat q\in
  Q\}$. 

Let $T=(Q,\Sigma,\rightarrow)$ be a triple
such that $Q$ is a set of elements called \emph{states}, $\Sigma$ is
an \emph{alphabet}, a set of elements called $letters$ or
\emph{labels}, and 
$\rightarrow\subseteq Q\times\Sigma\times Q$ is a \emph{transition
  relation} or \emph{set of transitions}. Then, $T$ is called a
\emph{Labelled Transition System (LTS)}. From 
$T$, given a letter $a\in\Sigma$, we define the two following relations on $Q$:
$\xrightarrow{a}\triangleq\{(q,q')\in Q\times Q\suchthat (q,a,q')\in
\rightarrow\}$ and its reverse
$\pre_{\xrightarrow{a}}\triangleq{(\xrightarrow{a})}^{-1}$. When
$\rightarrow$ is clear from 
the context, we simply note $\pre_a$ instead of
$\pre_{\xrightarrow{a}}$.
Finally, if for all $a\in\Sigma$, the relation $\xrightarrow{a}$ is a function
(i.e. $\forall a\in\Sigma\,\forall q\in Q\,.\,|{\xrightarrow{a}}(q)|\leq 1$)
then $T$ is said \emph{deterministic} and thus a DLTS.

\section{Underlying Theory}
\label{sec:underlyingTh}

The presentation of this section is quite different from what is found in
the literature. Thanks to the results it provides, I hope the reader will
be convinced of its pertinence.

Let $T=(Q,\Sigma,\rightarrow)$ be a LTS. The classical definition of a
simulation says that a relation $\mathscr{S}\subseteq Q\times Q$ is a
simulation over $T$ if for any 
transition $q_1\xrightarrow{a}q'_1$ and any state $q_2\in Q$ such that
$q_1\,\mathscr{S}\,q_2 $, there is a transition $q_2\xrightarrow{a}q'_2$
such that $q'_1\,\mathscr{S}\,q'_2$. However, the following definition
happens to be more effective than the classical one. As put in evidence by
the picture on the right, the two definitions are clearly equivalent.

\vspace{3ex}
\noindent
\begin{minipage}[c]{0.75\linewidth}
  \begin{definition}
    \label{def:sim}
    Let $T=(Q,\Sigma,\rightarrow)$ be a LTS and $\mathscr{S}$ be a relation
    on $Q$. The relation $\mathscr{S}$ is a \emph{simulation} over $T$ if:
    \begin{displaymath}
     \forall
    a\in\Sigma\;.\;\mathscr{S}\mathrel{\circ}\pre_a\subseteq\pre_a\mathrel{\circ}\mathscr{S}
    \enspace .
  \end{displaymath}
    For two states $q,q'\in Q$, we say ``$q$ is simulated by $q'$'' if there is a
    simulation $\mathscr{S}$ over $T$ such that $q\,\mathscr{S}\,q'$.
      \end{definition}
\end{minipage}
\hfill
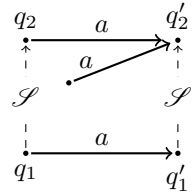
\begin{minipage}[c]{0.17\linewidth}
      \begin{tikzpicture}[baseline=(q3.south),shorten >=2pt,shorten
      <=2pt,font=\footnotesize]

      \path coordinate[label=below:$q'_1$] (q1) [fill] circle (1pt)
      (-2,0) coordinate[label=below:$q_1$] (q2)
      [fill] circle (1pt) (-2,1.5) coordinate[label=above:$q_2$] (q3) [fill] circle (1pt)
      (0,1.5) coordinate[label=above:$q'_2$] (q4) [fill] circle (1pt)
      (q3) 
      +(-45:.8) coordinate (q6) [fill] circle (1pt)
      ;
    
      \path[every edge/.style={->,dashed,draw},circle,inner sep=1pt, every node/.style={fill=white}]
      (q2) edge node (S) {$\mathscr{S}$} (q3) (q1) edge node
      {$\mathscr{S}$} (q4);

      \path[->,auto,circle,inner sep=1pt,thick] (q2) edge node {$a$} (q1)
      (q3)edge[shorten >=6pt] node (base) {$a$} (q4)
      (q6)edge[shorten >=3pt] node[near start] {$a$} (q4) ;
    \end{tikzpicture}
\end{minipage}
\bline

 A \emph{bisimulation} $\mathscr{S}$ is just a simulation such that $\mathscr{S}^{-1}$
is also a simulation.
For a bisimulation $\mathscr{S}$, two states $q$ and
$q'$ such that $q\,\mathscr{S}\,q'$ are said \emph{bisimilar}.

The main idea to obtain efficient algorithms is to consider relations
between blocks of states and not merely relations between
states. Therefore, we need a characterization of the notion of bisimulation
expressed over blocks.

\begin{proposition}
  \label{prop:blockBis}
  Let $T=(Q,\Sigma,\rightarrow)$ be a LTS and
  $\mathscr{S}$ be an equivalence relation on $Q$. The relation
  $\mathscr{S}$ is a bisimulation over $T$ if and only if for all block $B$ of
  $\mathscr{S}$ we have: $\forall a\in\Sigma\;.\;
  \mathscr{S}\circ\pre_a(B)\subseteq\pre_a(B)$. 
\end{proposition}
\begin{proof}
  
  If $\mathscr{S}$ is a bisimulation, by definition it is a
  simulation. Then, we have for any $B\subseteq Q$: $\forall
  a\in\Sigma\;.\;\mathscr{S}\mathrel{\circ}\pre_a(B)\subseteq\pre_a(\mathscr{S}(B))$. This  
  inclusion is thus also true for a block $B=\mathscr{S}(q)$ for a
  given $q\in Q$.  With the transitivity of an equivalence relation, we get
  $\mathscr{S}(B)=\mathscr{S}\circ\mathscr{S}(q)=\mathscr{S}(q)=B$. All
  of this implies:
  $\forall a\in\Sigma\;.\;\mathscr{S}\circ\pre_a(B)\subseteq \pre_a(B)$.
  
  In the other direction, let us consider a given $q\in Q$. Let
  $B=\mathscr{S}(q)$. An equivalence relation is
  reflexive. Therefore, we get: 
  $q\in B$. Then,
  \begin{displaymath}
     \mathscr{S}\mathrel{\circ}\pre_a(q)\subseteq
     \mathscr{S}\mathrel{\circ}\pre_a(B)\subseteq
     \pre_a(B)=
      \pre_a\mathrel{\circ}\mathscr{S}(q)\enspace .
  \end{displaymath}
Which implies that $\mathscr{S}$ is a simulation. Since $\mathscr{S}$ is an
equivalence relation we have $\mathscr{S}=\mathscr{S}^{-1}$ and thus
$\mathscr{S}^{-1}$ is also a simulation. This ends the proof. 
\end{proof}

 Let $\mathscr{R}$ be an
 equivalence relation on $Q$. Thanks to the preceding proposition, if
 $\mathscr{R}$ is not a bisimulation over $T$, there is $a\in\Sigma$ and $B$ a block
 of $\mathscr{R}$ such that
 $\mathscr{R}\circ\pre_a(B)\not\subseteq\pre_a(B)$. This implies the 
existence of a subset of $Q$, let us call it $\remove_a(B)$, such that: 
\begin{equation}
  \label{eq:withRemove}
\mathscr{R}\circ\pre_a(B)\subseteq\pre_a(B)\cup\remove_a(B)\enspace .
\end{equation}
The problem with \eqref{eq:withRemove} is that
$\remove_a(B)$ depends on a letter which we want to avoid in order to
obtain an algorithm whose complexity does not depend on the size of the
alphabet. It would therefore be more interesting to have something like: 
\begin{equation}
  \label{eq:withNotRel}
\mathscr{R}\circ\pre_a(B)\subseteq\pre_a(B\cup\notRel(B))\enspace .
\end{equation}
But, this is possible only if $\mathscr{R}\circ\pre_a(B)\subseteq\pre_a(Q)$. A
sufficient condition is: 
\begin{equation}
 \label{eq:InitRefineRestriction}  
 \forall a\in\Sigma\;.\;\mathscr{R}\circ\pre_a(Q)\subseteq\pre_a(Q)\enspace .
\end{equation}

The fact is that \eqref{eq:InitRefineRestriction} is not a real
restriction because any bisimulation included in $\mathscr{R}$ satisfies
this condition: a state which has an outgoing transition labelled by a
letter $a$ can be bisimilar only with states which have at least one
outgoing transition labelled by $a$. This is exactly what
\eqref{eq:InitRefineRestriction} says about $\mathscr{R}$. Surprisingly,
all of this is well known, but we have found no algorithm that uses this
as an optimization during an initialization phase. For the present paper
this not an optimization, this is a necessity. 

The next definition and lemma establish that we can restrict
our problem of finding the coarsest bisimulation inside an equivalence
relation $\mathscr{R}$ to the problem of finding the coarsest bisimulation
inside an equivalence relation $\mathscr{R}$ that satisfies
\eqref{eq:InitRefineRestriction}.
\begin{definition}
    Let $T=(Q,\Sigma, \rightarrow)$ be a LTS and 
   $\mathscr{R}$ be an equivalence relation on $Q$. We
   define $\initRefine(\mathscr{R})\subseteq \mathscr{R}$ such that:
   \begin{displaymath}
    (q,q')\in\initRefine(\mathscr{R}) \Leftrightarrow
    (q,q')\in\mathscr{R}\,\wedge\,
    \forall a\in\Sigma\;(q\in\pre_a(Q) \Leftrightarrow q'\in\pre_a(Q))\enspace .
   \end{displaymath}
\end{definition}

\begin{lemma}
  \label{lem:InitRefine}
  Let $T=(Q,\Sigma, \rightarrow)$ be a LTS and
  $\mathscr{U} = \initRefine(\mathscr{R})$ with
  $\mathscr{R}$ an equivalence relation on $Q$. Then, $\mathscr{U}$ is an
  equivalence relation on $Q$ and all bisimulation over 
  $T$ included in $\mathscr{R}$ is included in $\mathscr{U}.$
\end{lemma}
\begin{proof}\mbox{}
  \begin{itemize}
\item Since $\mathscr{R}$ is an equivalence relation and thus reflexive,
  $\mathscr{U}$ is also trivially reflexive. Its definition being
  symmetric, $\mathscr{U}$ is also symmetric. Now, let us suppose
  $\mathscr{U}$ is not transitive. There are 
  three states $q_1,q_2,q_3\in Q$ such 
  that: $q_1\,\mathscr{U}\,q_2\wedge q_2\,\mathscr{U}\,q_3 \wedge \neg\;
  q_1\,\mathscr{U}\,q_3$. From the fact that
  $\mathscr{U}\subseteq\mathscr{R}$ and $\mathscr{R}$ is an equivalence relation, we
  get $q_1\,\mathscr{R}\,q_3$. With $\neg\; q_1\,\mathscr{U}\,q_3$ and the
  definition of $\mathscr{U}$ there is $a\in\Sigma$ such that only one of $\{q_1,q_3\}$ 
  belongs to $\pre_a(Q)$. Let us suppose we have 
  $q_1\in\pre_a(Q)$ and $q_3\not\in\pre_a(Q)$. The problem is that $q_1\in\pre_a(Q)$ and
  $q_1\,\mathscr{U}\,q_2$ implies $q_2\in\pre_a(Q)$. With
  $q_2\,\mathscr{U}\,q_3$ we also get $q_3\in\pre_a(Q)$ which contradicts
   $q_3\not\in\pre_a(Q)$.
 \item Let us suppose the existence of a bisimulation $\mathscr{S}$
   included in $\mathscr{R}$ but not in $\mathscr{U}$. There are two states $q_1,q_2\in Q$ such
  that: $q_1\,\mathscr{S}\,q_2 \wedge \neg\;q_1\,\mathscr{U}\,q_2$. From
 $\mathscr{S}\subseteq\mathscr{R}$ we get $q_1\,\mathscr{R}\,q_2 $. With
 $\neg\; q_1\,\mathscr{U}\,q_2$ and the definition of $\mathscr{U}$ there
 is $a\in\Sigma$ such that only one of $\{q_1,q_2\}$ 
  belongs to $\pre_a(Q)$. Let us suppose we have  $q_1\in\pre_a(Q)$ and
  $q_2\not\in\pre_a(Q)$. With $q_1\,\mathscr{S}\,q_2$ we get $q_2\in
 \mathscr{S}\mathrel{\circ} \pre_a (Q)$. With the hypothesis that $\mathscr{S}$ is a
 simulation, we get $q_2\in \pre_a\mathrel{\circ} \mathscr{S}(Q)$ and thus
 $q_2\in \pre_a(Q)$, since $\mathscr{S}(Q)\subseteq Q$, which contradicts
 $q_2\not\in\pre_a(Q)$.
\end{itemize}
\end{proof}

The conjunction of \eqref{eq:InitRefineRestriction} and of the trivial
condition $\mathscr{R}(Q)=Q$ is the cornerstone of what follows.

\begin{definition}
   Let $T=(Q,\Sigma,\rightarrow)$ be a LTS and
  $\mathscr{R}$ be an equivalence relation on $Q$.
  \begin{itemize}\item 
    A subset $L$ of $Q$ is said a \emph{potential splitter} of
    $\mathscr{R}$ if $\mathscr{R}(L)=L$, $L$ is composed of at least two
    blocks of $\mathscr{R}$ and:
    $\forall a\in\Sigma\;.\;\mathscr{R}\circ\pre_a(L)\subseteq\pre_a(L)$.
  \item a couple $(L_1,L_2)$ of two non empty subsets of $Q$ is said a \emph{splitter} of
    $\mathscr{R}$ if: $L_1\cap L_2=\emptyset$, $\mathscr{R}(L_1)=L_1$, $\mathscr{R}(L_2)=L_2$
    and: $\forall a\in\Sigma\;.\;\mathscr{R}\circ\pre_a(L_1 \cup
    L_2)\subseteq\pre_a(L_1\cup L_2)$.
  \end{itemize}
\end{definition}

The key element of the preceding definition is that a splitter does not
depend on a specific letter. This is in sharp contrast to what is generaly
found in the literature since \cite{Hop71} (to our knowledge, there is only
one exception:
\cite{BC08}). The second difference with the 
literature is that our splitters are couples of 
sets of blocks. Indeed, in the literature, the second element of our
splitters is hidden.  However, its presence is essential as it allows us to split
sequentially, with different letters,
the current relation with the same splitter. Furthermore, it will give us
more freedom in the choice of the splitter to use.

From now on, we consider that $T=(Q,\Sigma, \rightarrow)$ is a DLTS,
$\mathscr{R}$ is an equivalence relation on $Q$ and $(L_1,L_2)$ is a splitter of
$\mathscr{R}$. We define for $X\subseteq Q$:
  \begin{displaymath}
    \splitBlock(X,\mathscr{R})\triangleq\mathscr{R}\setminus
    \cup_{q\in X}\{(q,q'),(q',q)\in Q\times Q\suchthat
    q'\in\mathscr{R}(q)\wedge q'\not\in
    X\}\enspace .
  \end{displaymath}

$ \splitBlock(X,\mathscr{R})$ amounts to split all blocks $C$ such that
$C\cap X\not=\emptyset$ and $C\setminus X\not=\emptyset$ in two blocks
$C_1=C\cap X$ and $C_2=C\setminus X$, see Fig.~\ref{fig:split}. 
When in an equivalence relation $\mathscr{R}$ we just split some blocks $C$
in two parts, the resulting 
relation is still an equivalence relation. Therefore,
$\splitBlock(X,\mathscr{R})$ is still an equivalence relation. Furthermore,
a block of $\mathscr{R}$ whose all elements are in $X$ or for which
no element is in $X$ is still a block in $\splitBlock(X,\mathscr{R})$. 

\begin{figure}[h]
  \centering
\subfigure[{From $\mathscr{R}$ to $\mathscr{U}$.}]{
  \begin{tikzpicture}[shorten >=2pt, shorten <=2pt,font=\footnotesize]
        \begin{scope}[every node/.style={draw,dashed,circle,minimum
            width=.5cm},node distance=7mm]
          \path  node (L11) {}
          node[below of=L11] (L12) {}
          node[below of=L12] (L13) {};
          \node [node distance=8mm,above of=L11,fill=black!20] (L21) {}
          node[above of=L21] (L22) {} ;
        \end{scope}
        \splitnodend[0]{}{fill=black!20}{L11}
        \splitnodend[0]{}{fill=black!20}{L12}
        \splitnodend[0]{}{fill=black!20}{L22}
        \node[draw,rounded corners,inner sep=3pt,fit=(L11) (L12)
        (L13),label=right:$L_1$] (L1) {};
        \node[draw,rounded corners,inner sep=3pt,fit=(L21)
        (L22),label=right:$L_2$] (L2) {};        
        \path
         (L1.north west) coordinate (nw1)
        (L2.south west) coordinate (sw1) 
       ($(nw1)!.5!(sw1)$) 
        ++(-1cm,0)
        node[circle,draw,dashed,minimum width=1cm,anchor=center,label=left:$C$] (C) {};
        \splitnodend[0]{}{fill=black!20}{C}

        \path (CUn.center) +(.2,0) coordinate[label={left}:{$q$}] (q) ;
        \path (CDeux.center) +(.2,0) coordinate[label={left}:{$q'$}] (q') ;
         \fill (q) circle (1pt) (q') circle (1pt) ;

        \begin{scope}[every node/.style={circle,minimum
            width=.5cm},node distance=7mm]
          \path ++(3.5cm,0) node (L11') {}
          node[below of=L11'] (L12') {}
          node[below of=L12',draw,dotted,semithick] (L13') {};
          \node [node distance=8mm,above of=L11',draw,dotted,semithick,fill=black!20] (L21') {}
          node[above of=L21'] (L22') {} ;
        \end{scope}
        \splitnode[1]{dotted,semithick}{dotted,semithick,fill=black!20}{L11'}
        \splitnode[1]{dotted,semithick,}{dotted,semithick,fill=black!20}{L12'}
        \splitnode[1]{dotted,semithick,}{dotted,semithick,fill=black!20}{L22'}
        \node[draw,rounded corners,inner sep=3pt,fit=(L11') (L12')
        (L13'),label=right:$L_1$] (L1') {};
        \node[draw,rounded corners,inner sep=3pt,fit=(L21')
        (L22'),label=right:$L_2$] (L2') {};       
        \path
        ($(L1'.north west)!.5!(L2'.south west)+(-1cm,0)$) 
        node[circle,minimum width=1cm,anchor=center] (C') {};
        \splitnode[2]{dotted,semithick}{dotted,semithick,fill=black!20}{C'}
        \path (C'Un.center) +(.2,0) coordinate[label={left}:{$q$}] (qR) ;
        \path (C'Deux.center) +(.2,0) coordinate[label={left}:{$q'$}] (q'R) ;
         \fill (qR) circle (1pt) (q'R) circle (1pt) ;

        \path (C'Un.south west) node[anchor=north,font=\scriptsize] {$C_1$};
        \path (C'Deux.north west) node[anchor=south,font=\scriptsize] {$C_2$};        
        \path[->,auto,circle,inner sep=1pt,thick]
        (q) edge node {$a$} (L1)
        (q') edge node {$a$} (L2)
        (qR) edge node {$a$} (L1')
        (q'R) edge node {$a$} (L2') ;    
      \end{tikzpicture}
      \label{fig:split}
} \hfill
\subfigure[center][Contradiction of \endgraf  $\mathscr{S}\subseteq\mathscr{R}$ and
$\mathscr{S}\not\subseteq\mathscr{U}$.]{
  \begin{tikzpicture}[shorten >=2pt, shorten <=2pt,font=\footnotesize]

    \path[fill] coordinate (q) circle (1pt) ;
    \path (q) node[left=3pt,inner sep=0] (qlabel) {$q$};
    \path[fill] (q) +(0,1.2) coordinate  (q') circle (1pt) ;
    \path (q') node[left=2pt,inner sep=0] (q'label) {$q'$};
    
    \node[draw,dashed,ellipse, minimum width=1cm, fit=(q) (qlabel) (q') (q'label)] (C)
    {} ;

    \path[every edge/.style={->,dashed,draw},circle,inner sep=0pt, every node/.style={fill=white}]
    (q) edge[bend angle=25,bend left] node[inner sep=0pt,font=\scriptsize] {$\mathscr{S}$} (q')
    ;

    \fill[nearly transparent] let \p{w}=(C.west), \p{n}=(C.north),
    \n{rh}={(\x{n} - \x{w})}, \n{rv}={(\y{n} - \y{w})} in 
    (C.east)  arc (0:-180:\n{rh} and \n{rv}) -- cycle;

    \path[fill](q) +(2,-1cm) coordinate (q1) circle (1pt) ;
    \path (q1) node[right=2pt,inner sep=0] (q1label) {$q_1$};
    \path[fill] (q1) +(0,1.2) coordinate  (q'1) circle (1pt) ;
    \path (q'1) node[right=2pt,inner sep=0] (q'1label) {$q'_1$};    
   
    \path[every edge/.style={->,dashed,draw},circle,inner sep=3pt, every node/.style={fill=white}]
    (q1) edge[bend angle=25,bend left] node[inner sep=0pt,font=\scriptsize] {$\mathscr{S}$} (q'1)
    ;

     \node[draw,dashed,ellipse,minimum width=1.3cm,fit=(q1) (q1label) (q'1) (q'1label)] (C')
    {} ;

    \path (C'.north) +(0,1) coordinate (nl1);
     \node[draw,rectangle,label=right:$L_1$,rounded corners,rounded corners=.5cm,fit=(C') (nl1) ] (L1)
    {} ;

     \path[->,auto,circle,inner sep=1pt,thick]
     (q) edge node {$a$} (q1)
       (q') edge node {$a$} (q'1) ;
   
  \end{tikzpicture}
  \label{fig:contradiction}
}
\captionsetup{justification=centering}
    \caption{
      $\mathscr{U}=\splitBlock
      (\pre_a(L_1),\mathscr{R})$.
    \endgraf \small Blocks of $\mathscr{R}$ are dashed, blocks
      of $\mathscr{U}$ are dotted, 
       $\pre_a(L_1)$ is in gray.}     
\label{fig:refinement}
  \end{figure}

Let $\mathscr{U}=\splitBlock(\pre_a(L_1),\mathscr{R})$ and let us consider
a block $C$ of $\mathscr{R}$ which contains an element $q$ in $\pre_a(L_1)$
and an element $q'$ not in $\pre_a(L_1)$, see Fig.~\ref{fig:split}. Since 
$C\xrightarrow{a}(L_1\cup L_2)$, by definition of a splitter, we have 
$C\subseteq\pre_a(L_1\cup L_2)$ and thus $q'$ is in $\pre_a(L_2)$. This implies that
$\splitBlock(\pre_a(L_1),\mathscr{R})=\splitBlock(\pre_a(L_2),\mathscr{R})$ and in no block of
$\mathscr{U}$ there is a state belonging in $\pre_a(L_1)$
and another state not belonging in $\pre_a(L_1)$. Therefore,
$\mathscr{U}\circ\pre_a(L_1)\subseteq\pre_a(L_1)$ and, by symmetry, 
$\mathscr{U}\circ\pre_a(L_2)\subseteq\pre_a(L_2)$. This is illustrated by
the right part of Fig.~\ref{fig:split}.

Now, let us suppose there is a bisimulation $\mathscr{S}$ included in
$\mathscr{R}$ but not included in $\mathscr{U}$, see
Fig.~\ref{fig:contradiction}. This means that there are two states $q$ and 
$q'$ of the same block of $\mathscr{R}$ such that $q\,\mathscr{S}\,q'$,
$q\in\pre_a(L_1)$ and $q'\not\in\pre_a(L_1)$. Let $q_1$ be the state of
$L_1$ such that $q\in\pre_a(q_1)$. We thus have
$q'\in\mathscr{S}\circ\pre_a(q_1)$. With the fact that $\mathscr{S}$ is a
simulation, we infer that $q'\in\pre_a\mathrel{\circ}\mathscr{S}(q_1)$
which implies the existence of a state $q'_1$ such that
$q_1\,\mathscr{S}\,q'_1$ and $q'\in\pre_a(q'_1)$. But, remember, we have
$\mathscr{S}\subseteq\mathscr{R}$, $q_1\in L_1$ and
$\mathscr{R}(L_1)=L_1$. Therefore, $q'_1\in L_1$ and
$q'\in\pre_a(L_1)$. Which contradicts $q'\not\in\pre_a(L_1)$.

From what precedes, we infer the following theorem.
\begin{theorem}
  \label{th:split}
  Let $T=(Q,\Sigma,\rightarrow)$ be a DLTS and $(L_1,L_2)$ be a
  splitter of $\mathscr{R}$ an equivalence relation on
  $Q$. Let $\mathscr{U}=\splitBlock(\pre_a(L_1),\mathscr{R})$. Then,
  $\mathscr{U}$ is an equivalence relation and any
  bisimulation included in  
  $\mathscr{R}$ is also included in $\mathscr{U}$. Furthermore:
  $\mathscr{U}\circ\pre_a(L_1)\subseteq\pre_a(L_1)$ and  
 $\mathscr{U}\circ\pre_a(L_2)\subseteq\pre_a(L_2)$.  
\end{theorem}

Let us come back to equation \eqref{eq:withRemove} and consider Fig.~\ref{fig:split}.
We have $\mathscr{R}\circ\pre_a(L_1)\subseteq\pre_a(L_1)\cup\remove_a(L_1)$
with $\remove_a(L_1)=\pre_a(L_2)$ which clearly depends on $a$. With
equation \eqref{eq:withNotRel} we 
have: $\mathscr{R}\circ\pre_a(L_1)\subseteq\pre_a(L_1\cup\notRel(L_1))$
with $\notRel(L_1)=L_2$ which does not depend on any letter. This was our goal.

From last theorem, the algorithm will maintain a set of potential
splitters. At each main iteration, we choose $L$ one of them. Then, we
choose in $L$ a block $B$, no matter which one. We set $L_1$ to be the
smallest set between $B$ and $L\setminus B$; we set $L_2$ to be the other
one. Then, we iteratively split the current partition, with the non empty
$\pre_a(L_1)$. The magical thing is that $(L_1,L_2)$ remains a splitter
during these iterations even if some blocks in $L_1$ or in $L_2$ are
split. Management of the letters being aside, the resulting algorithm is
quite similar to the one in \cite{PT87}. However, knowing that the LTS being
deterministic, there is no need of counters like in \cite{PT87}.

\section{The Bisimulation Algorithm }
\label{sec:bisim}

  \begin{function}
    \caption{Init($T,P_\mathrm{init}$) with $T={(Q,\Sigma,\rightarrow)}$\label{func:init}}   
    $P := \Copy(P_\mathrm{init})$;    
    \lForAll{$a\in\Sigma $\nllabel{init:l1}}
    {$a.\PreSmal :=\emptyset$;}

    \lForAll{$q\xrightarrow{a}q'\in\rightarrow$\nllabel{init:l2}}
    {                     
      $a.\PreSmal := a.\PreSmal \cup \{q\}$
    }     
                        
    \lForAll{$a\in \Sigma$\nllabel{init:l3}}
    {                
      $(P,\_) :=\Split(a.\PreSmal,P)$\;      
    }
    \Return{$(P)$}
  \end{function}
  
\begin{function}
  \caption{DBisim($T, P_\mathrm{init}$) with $T={(Q,\Sigma,\rightarrow)}$}
  \label{func:bisi}
  $P := \Init(T,P_\mathrm{init})$\nllabel{bisi:l1}\}\;
 \lIf {$|P|=1$}{\Return{$(P)$}} \tcc*[h]{nothing more has to
 be done}\nllabel{bisi:l2}\}\;
  $S:=\{Q\}$;
  $alph := \emptyset$;
  \lForAll{$a\in \Sigma $}
  {$a.\PreSmal := \emptyset$\nllabel{bisi:l3}}\;
             
  \While{$\exists L \in S$\nllabel{bisi:l4}} {      
    \tcc*[h]{Assert : $alph = \emptyset\wedge(\forall a\in\Sigma\,.\,a.\PreSmal
      = \emptyset)$\lnl{bisi:l5}}\;
      
    {Let $B$ be any block of $L$}\nllabel{bisi:l6}\;
 
    \lIf {there are only two blocks in $L$}{$S:=S\setminus\{L\}$}
    \lElse{$L:=L\setminus B$}\nllabel{bisi:l7}\;
    \lIf {$|B|\leq |L\setminus B|$}{$Smaller:=B$}
    \lElse{$Smaller:=L\setminus B$}\nllabel{bisi:l8}\;
    
    \ForAll{$q\xrightarrow{a}q_1\in \rightarrow$ \KwSty{such that} $q_1\in Smaller$\nllabel{bisi:l9}}
    {
      $alph := alph \cup \{a\}$;
      $a.\PreSmal := a.\PreSmal \cup \{q\}$\nllabel{bisi:l10}\;    
    }
      
    \ForAll{$a\in alph$ \nllabel{bisi:l11}}
    {        
      $(P, Splitted) :=
      \Split(a.\PreSmal,P)$\nllabel{bisi:l12}\;
        
      \ForAll{$C\in Splitted$\nllabel{bisi:l13}}
      {
        \lIf{$C\not\subseteq\cup S$}{$S:=S\cup\{C\}$\nllabel{bisi:l14}\;}
      }
    }    
    \lForAll{$a\in alph $\nllabel{bisi:l15}}{$a.\PreSmal := \emptyset$}\;
    $alph := \emptyset$\nllabel{bisi:l16}\;
  }
  $P_\mathrm{bis}:=P$;
  \Return{$(P_\mathrm{bis})$}
\end{function}

In the remainder of the paper, all LTS are finite and deterministic.
Given a DLTS $T=(Q,\Sigma,\rightarrow)$ and an initial
partition $P_\mathrm{init}$ of $Q$, inducing an equivalence relation
$\mathscr{R}_\mathrm{init}$, the algorithm manages a set $S$ of 
potentials splitters to iteratively refine the partition $P$ initially
equals to $P_\mathrm{init}$. At the end,
$P$ represents $P_\mathrm{bis}$ the partition whose induced equivalence relation
$\mathscr{R}_\mathrm{bis}$ is the coarsest bisimulation over $T$ included in
$\mathscr{R}_\mathrm{init}$.

The partition $P$ is a set of blocks. A block is assimilated with its set
of states. The set $S$ is a set of subsets of $Q$. We will see that $S$ is
a set of potential splitters. To each letter $a\in\Sigma$ we associate a
set of states noted $a.\PreSmal$ since it corresponds
to $\pre_a(Q)$ after
the \textbf{forall} loop at line~\ref{init:l2} of \texttt{Init} and 
to $\pre_a(Smaller)$ after
the \textbf{forall} loop at line~\ref{bisi:l9} of \texttt{DBisim}.

The main function \texttt{DBisim} uses two others functions: \texttt{Split},
and \texttt{Init}. Function
\texttt{Split}($X,P$) splits each block $C$ of $P$, having at least one
element in $X$ and another one not in $X$, in two blocks $C_1=C\cap X$ and
$C_2=C\setminus X$, and returns the resulting partition and the list of
blocks $C$ that have been split. It is mainly an implementation of function
$\splitBlock(X,\mathscr{R}_P)$ seen in the previous section. 
Function \texttt{Init}, also uses \texttt{Split} and returns the partition
whose induced equivalence relation is $\initRefine(\mathscr{R}_\mathrm{init})$. 

\subsection{Correctness}
\label{sec:correctness}

Let us first consider function \texttt{Init($T,P_\mathrm{init}$)}. All line
numbers in this paragraph refer to function 
\texttt{Init}. At line~\ref{init:l2}, we identify for each
$a\in\Sigma$ the states which have an outgoing transition labelled by
$a$. Then, at line~\ref{init:l3}, for each $a\in\Sigma$ we separate in all
the blocks the states which have an outgoing transition labelled by $a$ and
the states that do not have an outgoing transition labelled by $a$.  Clearly,
the resulting partition corresponds to the relation $\initRefine(R_\mathrm{init})$.

Let us consider function \texttt{DBisim$(T,P_\mathrm{init})$}. From now on, all
line numbers refer to function \texttt{DBisim}.
Let $S'=\{B\in P\suchthat B\not\subseteq \cup S\}$. We prove, by an induction,
that the following property is an invariant of the \textbf{while} loop of
\texttt{DBisim}. The relation $\mathscr{R}$ is the  
equivalence relation induced by the current partition $P$.
\begin{gather}
  \label{eq:S_isStable}
    L\in S \Rightarrow \mathscr{R}(L)=L \text{ and } L \text{ contains at
      least two blocks of $\mathscr{R}$}\enspace .
\end{gather}

 Just before the execution of the \textbf{while} loop, $S$
contains only one element: $Q$. Thanks to the test at line~\ref{bisi:l2},
$Q$ is made of at least two blocks, and we obviously have
$\mathscr{R}(Q)=Q$, which satisfies 
property \eqref{eq:S_isStable}.

Let us assume \eqref{eq:S_isStable} is satisfied before an iteration
of the \textbf{while} loop. The set $S$ can be modified only at lines
~\ref{bisi:l7} and \ref{bisi:l14}. Let $L$ be the element of $S$ chosen at
line~\ref{bisi:l4}. By induction hypothesis, $L$ is composed of at least
two blocks and $\mathscr{R}(L)=L$. If $L$ is withdrawn from $S$, at
line~\ref{bisi:l7}, this is
because $L$ is composed of 
exactly two blocks that are implicitly added in $S'$. If $L$ is not
withdrawn from $S$ then a block, $B$, 
is withdrawn from $L$ which is composed, induction hypothesis and condition
of the \textbf{if} at
line~\ref{bisi:l7}, of at least three blocks. From the hypothesis that
$\mathscr{R}(L)=L$ and from $\mathscr{R}(B)=B$ since $B$ is a block, we also have
$\mathscr{R}(L\setminus B)=L\setminus B$. Property \eqref{eq:S_isStable} is therefore
not modified by line~\ref{bisi:l7}. At line~\ref{bisi:l14}, a block $C$ which
has been split at line~\ref{bisi:l12}, and thus implicitly withdrawn from
$S'$, is added into $S$. Since $C$ has just been split, and thus contains
two complete blocks, property \eqref{eq:S_isStable} is still true and thus is an
invariant of the \textbf{while} loop.

Now, let us consider the two following properties:
\begin{gather}
   \label{eq:bisimS}
 \mathscr{S} \text{ is a bisimulation over $T$ included in }
 \mathscr{R}_\mathrm{init} \Rightarrow 
 \mathscr{S}\subseteq\mathscr{R} \enspace .\\
 \label{eq:SS'}
     L\in  S\cup S' \Rightarrow \forall a\in\Sigma\,.\,
     \mathscr{R}\circ\pre_a(L)\subseteq\pre_a(L)\enspace .
\end{gather}
Thanks to function \texttt{Init} and Lemma~\ref{lem:InitRefine}, these
properties are satisfied before the 
\textbf{while} loop. Remember that $S=\{Q\}$ and $S'=\emptyset$ at this
moment. 
Let us 
suppose they are satisfied before an iteration of the \textbf{while} loop. Let $L$ be the
element of $S$ chosen at line~\ref{bisi:l4}. After line~\ref{bisi:l8}, $L=L_1\cup
L_2$ with $L_1=Smaller$ and $L_2=L\setminus Smaller$. Since $Smaller$ is a
block and $\mathscr{R}(L)=L$ by \eqref{eq:S_isStable}, we have:
$\mathscr{R}(L_1)=L_1$
and $\mathscr{R}(L_2)=L_2$. This implies that $(L_1,L_2)$ is a
splitter of $\mathscr{R}$. Note that, during the iteration, $\mathscr{R}$
can only be refined (by function $\Split$, line~\ref{bisi:l12}). This
implies that $(L_1,L_2)$ stays a splitter of $ \mathscr{R}$ during the
iteration. Furthermore, after each iteration of the \textbf{forall} loop at
line~\ref{bisi:l11}, from Theorem~\ref{th:split}, we have
$\mathscr{S}\subseteq\mathscr{R}$ since  this is the case, by induction
hypothesis, before the iteration. The fact that, $\mathscr{R}$
can only be refined also implies that property~\eqref{eq:SS'} is still
satisfied after the iteration for all elements of $S\cup S'$ different from
$L_1$ and $L_2$. Let us consider their cases. Let $\mathscr{R}_c$ be the
value of $\mathscr{R}$ after the iteration of the \textbf{forall} loop at
line~\ref{bisi:l11} for $a=c$. Then, from Theorem~\ref{th:split} we have
$\mathscr{R}_c\circ\pre_c(L_1)\subseteq\pre_c(L_1)$ and
$\mathscr{R}_c\circ\pre_c(L_2)\subseteq\pre_c(L_2)$. At the
end of the iteration of the while loop, we obviously have
$\mathscr{R}\subseteq\mathscr{R}_c$ and thus: $\forall a\in
alph\,.\,(\mathscr{R}\circ\pre_a(L_1)\subseteq\pre_a(L_1) \wedge
\mathscr{R}\circ\pre_a(L_2)\subseteq\pre_a(L_2))$. Let $a\not\in alph$ this
means that $\pre_a(L_1)=\pre_a(Smaller)=\emptyset$ and thus
$\mathscr{R}\circ\pre_a(L_1)\subseteq\pre_a(L_1)$. Therefore, we have, with
$\mathscr{R}'$ the value  of $\mathscr{R}$ before the
iteration:
$\mathscr{R}'\circ\pre_a(L_2)=\mathscr{R}'\circ\pre_a(L)$ which is included,
by induction hypothesis of \eqref{eq:SS'}, in $\pre_a(L)=\pre_a(L_2)$ and thus
$\mathscr{R}\circ\pre_a(L_2)\subseteq\pre_a(L_2)$ at the end of the
iteration. In summary, properties \eqref{eq:bisimS} and \eqref{eq:SS'} are
invariants of the \textbf{while} loop.

We postpone to the next section the proof of termination of
$\mathtt{DBisim}$. This will be done by a
complexity argument. For the moment, note that 
the execution of the \texttt{while} loop ends when $S=\emptyset$ and
thus, by definition of $S'$, $P$ as a set of blocs is included in
$S'$. With \eqref{eq:SS'}, all of that implies, with $\mathscr{R}_\mathrm{bis}$ the
last value of $\mathscr{R}$:
$\forall B\in P\,\forall a\in\Sigma\;.\;
  \mathscr{R}_\mathrm{bis}\circ\pre_a(B)\subseteq\pre_a(B)$.
By Proposition~\ref{prop:blockBis} this means that $\mathscr{R}_\mathrm{bis}$ is a
bisimulation, by \eqref{eq:bisimS} it is the coarsest one included in
$\mathscr{R}_\mathrm{init}$.  Furthermore, since
$\mathscr{R}_\mathrm{bis}=\mathscr{R}_{P_\mathrm{bis}}$ with
$P_\mathrm{bis}$ a partition, the last value of $P$, it is an equivalence relation.

\subsection{Complexities}
\label{sec:complexities}
Let $X$ be a set of elements, we qualify an encoding of $X$ as
\emph{indexed} if the elements of $X$ are encoded in an array of 
$|X|$ slots, one for each element. Therefore, an element of $X$ can be
identified with its index in this array.
Let $T=(Q,\Sigma, \rightarrow)$ be a LTS, an encoding of $T$ is said
\emph{normalized} if the encodings of $Q$, $\Sigma$ and $\rightarrow$ are
indexed, a transition is encoded by the index of its source
state, the index of its label and the index of its destination state, and
if  $|Q|$ and $|\Sigma|$ are in $O(|{\rightarrow}|)$.
If $|\Sigma|$ is not in $O(|{\rightarrow}|)$, we
can restrict $\Sigma$ to its 
really used part $\Sigma'=\{a\in\Sigma\suchthat \exists q,q'\in
Q\;.\;q\xrightarrow{a}q'\in\rightarrow\}$ whose size is less than 
$|{\rightarrow}|$. To do this, we can use hash table techniques, sort
the set $\rightarrow$ with the keys being the letters labelling the transitions,
or more efficiently use a similar technique of the one we used in the algorithm to
distribute a set of transitions relatively to its labels (see, as an
example, the \textbf{forall} loop at line \ref{bisi:l9} of
\texttt{DBisim}). This is done in $O(|\Sigma|+|{\rightarrow}|)$ time and uses
$O(|\Sigma|)$ space. We learned that
this may be done in $O(|{\rightarrow}|)$ time, still with
$O(|\Sigma|)$ space, by using a technique presented
in \cite{VL08} and which is also called "weak sorting" according to
\cite{BC08}. 
If $|Q|$ is not in $O(|{\rightarrow}|)$ this means there are states that
are not involved in any transition. In general, these states are
ignored.  In fact, just after our initialization phase done by function
\texttt{Init}, these states are in blocks that will
not changed during the execution of the algorithm. Therefore, we can also
restrict $Q$ to its useful part $\{q\in Q\suchthat \exists q'\in Q\;\exists
a\in \Sigma\;.\; q\xrightarrow{a}q'\in\rightarrow \vee
\;q'\xrightarrow{a}q\in\rightarrow\}$ whose size is in
$O(|{\rightarrow}|)$. This is done like for $\Sigma$. 
\begin{center}
  All encodings of LTS in this section are assumed to be normalized.
\end{center}

Let us assume the following hypotheses:
\begin{itemize}
\item scanning the elements of one of the following sets is done in time
  proportional to its size: $\rightarrow$, $a.\PreSmal$ for any
  $a\in\Sigma$, $L$ an element of $S$, $B$ a block of $P$,
  the set of transitions leading to an element $L$ of $S$ or to a block $B$ of
  $P$, and $Splitted$.
\item $\Split(a.\PreSmal,P)$ is executed in time proportional to the size of
  $a.\PreSmal$.
\item function $\Copy(P_\mathrm{init})$ is executed in $O(|Q|)$ time.
\item all the other individual instructions in functions \texttt{Init} and
  \texttt{DBisim} are done in constant time or amortized constant time.
\item all the data structures use only $O(|{\rightarrow}|)$ space.
\end{itemize}

Let us consider the time complexity of function \texttt{Init}. The
\textbf{forall} loop at line \ref{init:l1} is done in time
$O(|{\Sigma}|)$ and thus $O(|{\rightarrow}|)$ since $T$ is supposed to be
normalized. The \textbf{forall} loop at line \ref{init:l2} is done in time 
$O(|{\rightarrow}|)$. The \textbf{forall} loop at line \ref{init:l3} is
also done in time $O(|{\rightarrow}|)$ since we have: 
$\Sigma_{a\in\Sigma}|a.\PreSmal|\leq
\Sigma_{a\in\Sigma}|{\xrightarrow{a}}|\leq
|{\rightarrow}|$. Therefore, \texttt{Init} is done in
$O(|{\rightarrow}|)$-time.

Let us now consider function \texttt{DBisim}. Let us first remark that
during an iteration of the \texttt{forall} loop at line~\ref{bisi:l11}, 
$|Splitted|\leq |a.\PreSmal|$. Furthermore, during an iteration of 
the \textbf{while} loop,
 $|alph|$ and $\Sigma_{a\in alph}|a.\PreSmal|$  are less than the
number of transitions scanned during the \texttt{forall} loop at
line~\ref{bisi:l9}. This implies that the time complexity of function 
\texttt{DBisim} is proportional to the overall number of transitions scanned at
line~\ref{bisi:l9}.
But thanks to line~\ref{bisi:l8}, each time a transition
$q\xrightarrow{a}q'$ is used at line~\ref{bisi:l9}, $q_1$ belongs to the
set $Smaller$
whose size is less than half of the size of the previous time.

From all of this, we get the following theorem.

\begin{theorem}
  \label{th:complexities}
  Let $T=(Q,\Sigma, \rightarrow)$ be a DLTS and $P_\mathrm{init}$ be an
  initial partition of $Q$ inducing an equivalence relation
  $\mathscr{R}_\mathrm{init}$. Function \texttt{DBisim} 
  computes 
  $P_\mathrm{bis}$ the partition whose corresponding equivalence relation
  $\mathscr{R}_\mathrm{bis}$ is the coarsest bisimulation over $T$ included in
  $\mathscr{R}_\mathrm{init}$ in: 
   $O(|{\rightarrow}|\log |Q|)$-time and 
   $O(|\Sigma| + |{\rightarrow}| + |Q|)$-space.
\end{theorem}

The $O(|Q| + |\Sigma|)$
part in the space complexity of the theorem is due to the
normalization of $T$ as explained at the beginning of this sub section.

\subsection{Implementation}
\label{sec:implementation}

The data structures can be similar to what is classically used in
minimization and bisimulation algorithms like those in
\cite{Hop71,PT87}. However we do prefer some ideas found in \cite{VL08} and
rediscovered in \cite{Cec12}. Instead of a doubly linked list to represent
the states of a block, both papers distribute the indexes of the states in
an auxiliary array $A$ such that states in a same block form a subarray
of $A$. Furthermore, when a block $C$ has to be split in $C_1$ and $C_2$,
the states of $C_1$ and $C_2$ stay in the same subarray of $A$
corresponding to $C$. The only modification is that states of $C_1$ are put
on the left side of that subarray and thus states of $C_2$ on the right side.
Therefore, to encode a block, we just have to memorize a left and a right
index. A key advantage for the present paper, is that to represent an
element $L$ of $S$ we also just have to memorize a left and a right index in
$A$. Therefore, when the blocks of $L$ are split, we do not have to update
these variables. The other elements are classic: each state knows the block
to which it belongs, each block and each element of $S$ maintains its size,
each block maintains a boolean to know whether it belongs to $\cup S$ (this
information is transmitted to the sub block when a block is split),
$S$ may be encoded by a file or a by a stack. To choose a block in $L$ at
line~\ref{bisi:l6} we just choose the left or the right block in the
subarray of $A$ corresponding to $L$. Thus, it is easy to perform the
instruction $L:=L\setminus B$ in line~\ref{bisi:l7} and to scan the
elements of $Smaller$. The transitions are also initially sorted by a
counting sort in order to have the transitions leading to the same state
form a subarray of the array of all the transitions. This allows us to scan the
transitions leading to a given state in time proportional to their number.

The variable $Splitted$ is used just for the
clarity of the presentation. Indeed, when a block $C$ not in one of the
potential splitters of $S$ is split, this is detected during the call of
function $\Split$ and
$C$ is directly added into $S$.

The implementation of function $\Split(X,P)$ is not original (see for
example \cite[Sect. 6]{Cec12} where a version which also returns the blocks in $X$ is
given). The time complexity of a call is $O(|X|)$ when the elements of
$X$ may be scan in time proportional to its size.

In \cite[Sect. 6]{Cec12} an encoding, satisfying the hypotheses, of $alph$ and the $a.\PreSmal$ is
given. However, their role is just 
to distribute the set of
transitions leading to $Smaller$ in function of their label. Then, for each
label found, we scan the corresponding transitions to do the split. The
distribution of the transitions can simply be done by a kind of counting
sort.  It runs in $O(|Smaller|)$-time and uses
$O(|{\rightarrow}|)$-space. The "weak
sorting" technique used in \cite{BC08} can also be used.

In summary, the data structures that we use satisfy the hypotheses given at the
second paragraph of Sect.~\ref{sec:complexities} under the
assumption of a normalized DLTS.

\section{Main Application and Future Works}
\label{sec:conclusion}

The main application of the coarsest bisimulation problem over
finite DLTS is the minimization of deterministic automata (LTS
with the precision of an initial state and a set of final states). It
is well known that, when there is no useless state (a state $q$ is useless
if there is no path from the
initial state to $q$ or if there is no path from $q$ to a final state) this amounts
of finding the coarsest bisimulation 
included in the following equivalence relation:
$\mathscr{R}_\mathrm{init}=\{(q,q')\in Q\times Q\suchthat q\in F \Leftrightarrow q'\in
F\}$ with $F\subseteq Q$ the set of final states. The blocks of this
relation are: $F$, if it is not empty (but in that case the minimal automata
is the empty one), and $(Q\setminus F)$ if it is also not empty. Therefore,
the algorithm presented in the present paper can be used to minimize a
deterministic automata with the complexities announced.

Although this was not the purpose, a good piece of news is that
\texttt{DBisim} is very similar to the algorithm in \cite{PT87}. The
main difference being that our LTS are deterministic. An extension
of the present paper to the coarsest bisimulation problem over LTS will be
done. This will yield an  algorithm with
the same time and space complexities  of \cite{Val09} but 
simpler.

\bibliography{simulation}

\end{document}